\documentclass[a4paper,11pt,twoside]{article}
\usepackage{amsmath,amssymb,epsf}
\usepackage{amsfonts,amsbsy,amscd,stmaryrd}
\usepackage{paralist}
\usepackage{latexsym,amsopn}
\usepackage{verbatim}
\usepackage{amsthm}
\usepackage{times}
\usepackage{geometry}
\usepackage{setspace}
\usepackage{fancyhdr}
\makeatletter
\renewcommand{\@maketitle}{
\newpage
 \null
 \vskip 1.6em%
 \begin{center}%
  {\LARGE \@title \par}
 \end{center}%
 \vskip 1.4em%
 \begin{center}%
  {\@author \par}%
 \end{center}%
 \par} \makeatother

\renewcommand{\ker}{{\rm Ker}}
\newcommand{\ran}{{\rm Ran}}
\newcommand{\Ker}{{\rm Ker}}
\newcommand{\Ran}{{\rm Ran}}
\newcommand{\bra}{\langle} 
\newcommand{\ket}{\rangle}

\renewcommand\Re{{\rm Re\,}}

\renewcommand{\i}{{\rm i}}
\renewcommand{\d}{{\rm d}}

\newcounter{resultcounter}[section]

\theoremstyle{definition}
\newtheorem{thm}{Theorem}[section]
\newtheorem{theorem}[thm]{Theorem}
\newtheorem{lemma}[thm]{Lemma}

\newtheorem{proposition}[thm]{Proposition}

\newtheorem{definition}[thm]{Definition}
\theoremstyle{remark}
\newtheorem{remark}[thm]{Remark}
\theoremstyle{definition}
\newtheorem{corollary}[thm]{Corollary}

\numberwithin{equation}{section}

\newcommand{\beq}{\begin{equation}}
\newcommand{\eeq}{\end{equation}}
\newcommand{\beqa}{\begin{eqnarray}}
\newcommand{\eeqa}{\end{eqnarray}}

\def\one{{\mathchoice {\rm 1\mskip-4mu l} {\rm 1\mskip-4mu l} {\rm
      1\mskip-4.5mu l} {\rm 1\mskip-5mu l}}}

\def\cD{{\mathcal D}}  
  \def\cI{{\mathcal I}}
\def\cJ{{\mathcal J}} \def\cK{{\mathcal K}} 
  
\def\cP{{\mathcal P}}  \def\cR{{\mathcal R}}
\def\cS{{\mathcal S}} \def\cT{{\mathcal T}}

\newcommand{\R}{{\mathbb R}}
\newcommand{\N}{{\mathbb N}}

\newcommand{\C}{{\mathbb C}}
\newcommand{\Z}{{\mathbb Z}}

\def\proof{\noindent{\bf Proof.}\ \ }
        {\begin{compactitem}[#1]}{\end{compactitem}}
\newenvironment{innerlist3}[1][\enskip\textbullet]%
        {\begin{compactitem}[#1]}{\end{compactitem}}

\newcommand{\supp}{{\rm supp}}

\newcommand{\sd}{{\rm sd}}
\newcommand{\sdo}{{\rm sd\,}}
\renewcommand{\sp}{{\rm sp}}

\def\cf{C^{\infty}}

\def\Dor{\mathcal{D}'(\{0\})_{\leq r}}

\def\Do{\mathcal{D}'(\{0\})}

\def\du{\dot{u}}

\def\bv{\mathbf{v}}

\def\ord{{\rm ord}}

\def\Toff{T_{\rm off}}
\def\Ton{T_{\rm on}}

\newcommand{\dal}{\Box}

        {\begin{compactitem}[#1]}{\end{compactitem}}

\newenvironment{myindentpar}[1]%
{\begin{list}{}%
         {\setlength{\leftmargin}{#1}}%
         \item[]%
}
{\end{list}}

\begin{document}

\title{On-shell extension of distributions}
\author{\textsc{Dorothea Bahns$^{1}$,\  Micha\l\ Wrochna$^{2}$}}

\date{}
\maketitle
\begin{center}\small
$^{1}$Courant Research Centre ``Higher Order Structures in Mathematics''\\
$^{2}$RTG ``Mathematical Structures in Modern Quantum Physics''\\
\vspace{0.1cm} $^{1,2}$Universit\"at G\"ottingen, Bunsenstr. 3-5, D - 37073 G\"ottingen, Germany\\
\vspace{0.2cm} e-mail: $^{1}$\texttt{bahns@uni-math.gwdg.de}, \ \ $^{2}$\texttt{wrochna@uni-math.gwdg.de} \vspace{0.5cm} \end{center}

\begin{abstract}We consider distributions on $\R^n\!\setminus\!\{0\}$ which satisfy a given set of partial differential equations and provide criteria for the existence of extensions to $\R^n$ that satisfy the same set of equations on $\R^n$. 

We use the results to construct distributions satisfying specific renormalisation conditions in the Epstein and Glaser approach to perturbative quantum field theory. Contrary to other approaches, we provide a unified apporach to treat Lorentz covariance, invariance under global gauge group and almost homogeneity, as well as discrete symmetries. We show that all such symmetries can be recovered by applying a linear map defined for all degrees of divergence.

Using similar techniques, we find a relation between on-shell and off-shell time-ordered products involving higher derivatives of the fields.
\end{abstract}

\vspace{0.5cm}

\setlength{\parindent}{0.6cm}

\section{Introduction}

\hspace{1cm}In the program of Epstein and Glaser \cite{epsteinglaser}, the renormalisation problem in  quantum field theory is reformulated as a problem of extending distributions defined  on $\R^n\setminus\{0\}=:\dot\R^n$ to distributions on $\R^n$. By construction, any two extensions can differ by a distribution supported in 0, and one way to constrain this ambiguity is to require that the extension should have the same scaling degree as the original distribution. In quantum field theory, the ambiguity is further constrained by imposing physically motivated renormalisation conditions. Such conditions include the requirement that if $u\in\cD'(\dot\R^n)$ respects a global symmetry, e.g. it is invariant under the Lorentz group or a global gauge group, then the same should hold for its extension to $\R^n$. Another such condition, which turned out to be essential  in renormalisation on both flat and on curved space-time  \cite{hollwald}, 
 is the requirement  that if $u$ is homogeneous, its extension should be homogeneous as well, or at least it should behave as much like  a homogeneous distribution as possible (a property which can be properly formulated in terms of almost homogeneous distributions).

A problem which at first sight seems to be unrelated to such renormalisation conditions occurs in the construction of on-shell time ordered products involving higher derivatives of quantum fields. Roughly speaking, one wishes to extend to $\R^n$ an expression in $\cD'(\dot\R^n)$ involving derivatives of the Feynman propagator and Heaviside theta functions such that the extension satisfies the (free) equation of motion. The possibility of finding such an extension can be rephrased using the relation between on-shell time-ordered products (ordinarily used in quantum field theory) and off-shell products, the latter of which have proved to be better suited for a theoretical study of Epstein-Glaser renormalisation \cite{DF03,DF04}. 

The main idea in our approach is that all these problems can be formulated and solved in a unified framework, by restating them in terms of the existence of extensions 
which solve a set of (differential) equations.
More precisely, we state the following extension problem:

\begin{myindentpar}{0.4cm}
Let $\{Q^i\}_{i=1}^{k}$ be a family of differential operators on $\R^n$ with smooth coefficients, and let $u$ be a distribution in $\cD'(\dot\R^n)$ that satisfies
\[
Q^i u =0 \quad \mbox{ on \ } \dot\R^n \quad (i=1,\dots,k).
\]
Find  $\dot u\in\cD'(\R^n)$ such that $\dot u = u$ on $\dot\R^n$ and $Q^i\dot u=0$ on $\R^n$ ($i=1,\dots,k$). If such extensions $\dot u$ exist, we call them on-shell extensions (w.r.t. $\{Q^i\}_{i=1}^{k}$).
\end{myindentpar}

Indeed, invariance of a distribution under the action of a connected Lie group is equivalent to it being 
a solution of its infinitesimal generators. (Almost) homogeneity is described using (powers of) the operator $\sum_i x_i\partial_i -a$. In the construction of on-shell time-ordered products, the differential operator of interest is the Klein-Gordon operator $(\Box+m^2)$. To include discrete symmetries, we will consider a more general class of  operators later. 

On the mathematical side, the `on-shell extension' problem we consider is closely related to the so-called Bochner's extension problem, an issue which we explain in the text.

\medskip 
One advantage of our reformulation is the following. The various constructions and prescriptions to implement renormalisation conditions proposed  so far in e.g. \cite{scharf,prange,improved1,improved2,DF04,hollwald}, see also \cite{pct}, 
are each limited to a particular type of symmetry. Therefore, the simultaneous implementation of a number of different conditions requires cumbersome proofs of compatibility. Our framework on  the other hand, allows for a compact formulation of e.g. sufficient conditions on the existence of extensions subject to different renormalisation conditions, such as Lorentz invariance considered together with almost homogeneity and (eventually) parity. Moreover, it exhibits a new feature of Epstein-Glaser renormalisation: a renormalisation condition corresponding to a (differential) operator $Q$ can be imposed by applying a \emph{linear} map to a generic extension (which is a solution of $Q$ on $\dot\R^n$). This statement extends to the case of several renormalisation conditions.

Concerning the relation between on-shell and off-shell time-ordered products, we would like to point out that it was given in \cite{DF03,DF04} in terms of a recurrence relation for which an explicit solution was found in \cite{onshell}. In our framework we find a more compact formula, which contrary to that given in \cite{onshell} does not contain unnatural combinatorial factors. Instead, the coefficients which appear in our formula are simply eigenvalues of certain finite-dimensional operators directly related to the Klein-Gordon operator.  

\medskip
The paper is organized as follows. In Section \ref{sec:preliminary} we recall the definition of the scaling degree and degree of divergence, and review basic facts on distributions and extensions. Section \ref{sec:main} contains the main ideas and results. First, in subsection \ref{sec:essord}, we introduce the notion of operators of essential order $m$ on $\cD'(\R^n)$. Such operators  generalize differential operators and will enable us   to treat discrete symmetries. In subsection \ref{sec:spaces} we equip the finite dimensional spaces of distributions of given maximal order supported at the origin with a scalar product. The restrictions of operators $Q:\cD'(\R^n)\to\cD'(\R^n)$ to these spaces play an important role in subsection 3.2, especially  for Theorem \ref{thm:core}. This theorem provides a  solution to the extension problem with respect to an operator $Q$ of arbitrary  essential order in the sense that it lists different statements which are equivalent to the existence of on-shell extensions, and provides a candidate for such an extension. This candidate can be calculated explicitly and only requires one to find the eigenvalues of a finite dimensional matrix. Special cases, examples and generalizations are then discussed. Of particular interest is Theorem~\ref{thm:criterion} which states sufficient, and easy-to-check conditions that ensure the existence of on-shell extensions w.r.t. an operator of essential order 0.  We then explain how the extension problem with respect to a finite number of operators is solved. Section \ref{sec:products} is devoted to the construction of on-shell time-ordered products involving higher derivatives of the fields. We clarify how the relation between the on-shell and the off-shell formalism can be formulated and understood in our framework. Only this last section requires knowledge of quantum field theory. An outlook is presented in Section \ref{sec:outlook}.

\section{Preliminaries}\label{sec:preliminary}

\subsection{Notations}

Throughout the paper, we use the notation $\dot\R^n=\R^n\setminus\{0\}$.

For a continuous operator $Q^{\rm t}:\cD(\R^n)\to\cD(\R^n)$, we denote its transpose by $Q: \cD'(\R^n)\to\cD'(\R^n)$, i.e. $\bra Qu,\varphi\ket=\bra u,Q^{\rm t}\varphi\ket$, $u\in\cD'(\R^n)$, $\varphi\in\cD(\R^n)$.

\subsection{Scaling degree and degree of divergence}

We start by recalling the definition of Steinman's scaling degree of a distribution \cite{steinmann}, a notion which has proved to be very useful in renormalised perturbation theory. Most notably, it allowed (after a suitable generalization) to extend the Epstein-Glaser approach to curved space-times \cite{BF00}. It  does not seem to have entered the mathematical literature as such, although similar estimates have been used in e.g. \cite{Estrada:Reg}, see also \cite{wavelets} and references therein. Note also that Steinman's scaling degree is a degenerate case of Weinstein's degree of a distribution \cite{weinstein}. 

Consider the natural action of the dilation group $\R_{>0}$ on $\cD(\R^n)$ and its dual on $\cD'(\R^n)$, i.e. for a distribution $u\in\cD'(\R^n)$ and $\lambda >0$ set  
\[
\bra u_{\lambda}, \varphi \ket := \lambda^{-n} \bra u, \varphi(\lambda^{-1}\cdot) \ket, \quad \varphi\in\cD(\R^n).
\]
\begin{definition} \label{def:sdegree}The \emph{scaling degree} of $u\in\cD'(\R^n)$, denoted $\sd\,u$, is the infimum over all $\omega\in\R$ s.t. $\lim_{\lambda\searrow0}\lambda^{\omega}u_\lambda=0$ in $\cD'(\R^n)$. The \emph{degree of divergence} of $u$ is $\deg u := \sd\, u -n$.
\end{definition}
The degree of divergence of a distribution $u\in\cD'(\dot\R^n)$, denoted also $\deg u$, is defined analogously (one simply replaces $\cD(\R^n)$ by $\cD(\dot\R^n)$ and $\cD'(\R^n)$ by $\cD'(\dot\R^n)$ in the above definition). The difference is that it is possible here that the limit $\lim_{\lambda\searrow0}\lambda^{\omega}u_\lambda$ does not exist for any $\omega\in\R$. In this case we write $\deg\, u=\infty$. 

We will use the degree of divergence rather than the scaling degree, as it is more convenient in our framework. As a basic example, observe that the derivatives of the $\delta$-distribution  $\delta^{(\alpha)}$ on $\R^n$ have degree of divergence $|\alpha|$. A  function in $\cf(\R^n)$, considered as an element of $\cD'(\R^n)$, has degree of divergence at most $-n$. A distribution which is homogeneous of degree $a\in\C$ on $\R^n$ (resp. $\dot\R^n$) has degree of divergence $-\Re a-n$ in $\cD'(\R^n)$ (resp. $\cD'(\dot\R^n))$.

Let us briefly recall the basic properties of the degree of divergence, which were proved\footnote{Strictly speaking, the third claim is considered there only for $k=0$, but the case $k\geq 1$ follows immediately by noting that under the assumptions, such $f$ equals $x^\beta g$ for some $g\in\cf(\R^n)$, $|\beta|=k$, and then using \ref{itsd2}.} in \cite[Lemma~5.1]{BF00}.

\begin{lemma}\label{lemmasd}Let $u\in\cD'(\R^n)$ and assume $\deg u <\infty$. Then:
\begin{enumerate}\setlength{\itemsep}{-1.5pt}
\item\label{itsd1} For $\alpha\in\N^n$, \ $\deg(\partial^{\alpha}u)\leq \deg u + |\alpha|$.
\item\label{itsd2} For $\alpha\in\N^n$, \ $\deg(x^{\alpha}u)\leq \deg u - |\alpha|$.
\item\label{itsd3} Let $f\in\cf(\R^n)$ and assume $f^{(\alpha)}(0)=0$ for $|\alpha|\leq k-1$ , then \ $\deg (f u) \leq \deg u-k$.
\item\label{itsd4} Let $v\in\cD'(\R^k)$ then  $\sd (u\otimes v)\leq \sdo u + \sdo v$.
\end{enumerate}
\end{lemma}

\subsection{Extension of distributions}

Let us recall the basic ingredients of the construction of extensions of distributions. Essentially, we follow \cite{BF00}, but for later purposes, we make systematic use of the following spaces of distributions.
Denote by $\Do$ the space of distributions supported at $\{0\}$. 
For $r\geq 0$, let $\Dor$ be the subspace of $\Do$
given by those of maximal degree $r$,
\[
\Dor={\rm span}\,\left\{ v\in\Do : \ \deg v\leq r \right\}
={\rm span}\,\left\{ \delta^{(\alpha)}\in\cD'(\R^n) : \ |\alpha|\leq r \right\}.
\]
On the other hand, consider the space of all test functions vanishing up to order $r$ at $x=0$:
\beq
\cD_{r}(\R^n):=\{ \varphi\in\cD(\R^n) : \  (\partial_x^{\alpha}\varphi)(0)=0 \ \ \forall \alpha\in\N_0^{n}, |\alpha|\leq r \}.
\eeq
It will be convenient to generalize this definition in the following way. Let $\cK$ be a finite dimensional subspace of $\cD'(\{0\})$. Set
\[
\cD_{\cK}(\R^n):=\{ \varphi\in\cD(\R^n) : \ \bra v,\varphi\ket=0 \ \ \forall v\in\cK \}.
\]
Clearly, $\cD_r(\R^n)$ equals $\cD_{\cK}(\R^n)$ with $\cK=\Dor$. Observe that the scaling degree of a distribution in $\cD'_{\cK}(\R^n)$ can be defined in an analogous way to the scaling degree in $\cD'(\R^n)$. We now restate Theorem 5.2 from \cite{BF00} as follows:

\begin{proposition} \label{lemmarenorm} Let $u\in\cD'(\dot\R^n)$ have degree of divergence $r:=\deg u<\infty$. Then it admits a unique extension $\tilde u \in \cD'_{r}(\R^n)$ with the same degree of divergence $r$, given by
\beq\label{def:cDr}
\bra \tilde u,\varphi \ket := \lim_{\rho\to\infty}\bra u , (1-\vartheta_{\rho}) \varphi \ket, \ \ \ \varphi\in\cD_{r}(\R^n)
\eeq
where $\vartheta_{\rho}(x):=\vartheta(2^{\rho} x)$ and $\vartheta$ is an arbitrary function in $\cD(\R^n)$ such that $\vartheta=1$ in a neighbourhood of the origin.
\end{proposition}

It now remains to find elements of  $\cD'(\R^n)$ which correspond to the extension $\tilde u \in \cD'_r(\R^n)$. Following the ideas of \cite{BF00}, we do so by considering projections\footnote{This means that $W^{\rm t}:\cD(\R^n)\mapsto \cD_r(\R^n)$ is continuous and $(W^{\rm t})^2=W^{\rm t}$.} $W^{\rm t}:\cD(\R^n)\to\cD_{r}(\R^n)$, and applying their transpose $W:\cD'_r(\R^n)\to\cD'(\R^n)$ to $\tilde u\in\cD'_r(\R^n)$.

To this end, let us first state two lemmas which are slight generalizations of results found in \cite{DF04} where they were stated for $\cK=\Dor$.

\begin{lemma} Let $\cK$ be a finite dimensional subspace of $\Do$, let $\{v_i\}_{i\in \cI}$ be a basis of $\cK$, and assume $\{\psi_{i}\}_{i\in \cI}$ is a family $\psi_i\in\cD(\R^n)$ s.t. $\bra v_i , \psi_j \ket=\delta_{ij}$. Then
\beq\label{defWt}
W^{\rm t}\varphi=\varphi-\sum_{i\in\cI}\bra v_i, \varphi\ket \psi_i
\eeq
defines a projection $W^{\rm t}:\cD(\R^n)\to\cD_{\cK}(\R^n)$. Conversely, if $W^{\rm t}:\cD(\R^n)\to\cD_{\cK}(\R^n)$  is a projection, there is a family $\{\psi_{i}\}_{i\in \cI}$ with the above properties.
\end{lemma}
\begin{lemma}\label{lem:Wpreserves}Let $u\in\cD'_{\cK}(\R^n)$ and let $W^{\rm t}:\cD(\R^n)\to \cD_{\cK}(\R^n)$ be a projection. Then $\bra Wu,\varphi\ket=\bra u,\varphi\ket$ for all $\varphi\in\cD_{\cK}(\R^n)$ and $\deg Wu=\deg u$.
\end{lemma}

Taking into account that $\deg\delta^{(\alpha)}=|\alpha|$, we find the following important result on the existence of extensions with the same degree of divergence.

\begin{corollary} (\cite{BF00}) Let $u\in\cD'(\dot\R^n)$ be a distribution with  $r:=\deg u<\infty$. Then there is an extension $\dot u\in\cD'(\R^n)$ of $u$ with $\deg\dot u=\deg u$. Each such extension can be written as $\dot u = W\tilde u$, where $\tilde u$ is the unique extension of $u$ in $\cD'_r(\R^n)$ and $W^{\rm t}:\cD(\R^n)\to\cD_r(\R^n)$ is a projection.  Moreover, two arbitrary extensions with the above properties differ by an element of $\Dor$. 
\end{corollary}

While each extension of $u$ with the same degree of divergence can be constructed as above by using a projection $W^{\rm t}$, it can sometimes be more convenient to use some other operator $V^{\rm t}$ which maps $\cD(\R^n)$ to $\cD_r(\R^n)$ and check whether $V\tilde u$ is an extension of $u$ with the correct degree of divergence. Indeed, this approach will prove to be more convenient in directly constructing on-shell  extensions (cf. Proposition \ref{prop:direct}).


\section{On-shell extension problem}\label{sec:main}

\subsection{Essential order}\label{sec:essord}

Our aim is to implement symmetries, i.e. we will ask our extensions to satisfy a set of given equations. In order to include discrete symmetries, we consider more general operators from $\cD'(\R^n)$ to $\cD'(\R^n)$ rather than just differential operators.

\begin{definition}We say that $Q:\cD'(\R^n)\to\cD'(\R^n)$ is an \emph{operator of essential order $q$} if
\begin{innerlist3}
\item $Q$ is the transpose of a linear operator $Q^{\rm t}:\cf(\R^n)\to\cf(\R^n)$ which continuously maps $\cD(\dot\R^n)$ and $\cD(\R^n)$ to themselves; 
\item $q\in\N_0$ is the lowest number such that $\deg Q u \leq \deg u + q$ for all $u\in\cD'(\R^n)$.
\end{innerlist3}
\end{definition}

Basic examples for such operators are of course differential operators, for which the essential degree was already considered in \cite{nikolov}: a differential operator of order $m$ has essential order smaller or equal $m$. More precisely, a differential operator $Q=\sum_{|\alpha|\leq m}a_{\alpha}(x)\partial^{\alpha}$ has essential order $q$, where $q$ is the smallest possible non-negative number s.t. $(\partial^{\beta }a_{\alpha})(0)=0$ for $|\beta|\leq |\alpha|-q-1$. In particular, $Q$ has essential order $0$ if $(\partial^{\beta }a_{\alpha})(0)=0$ for $|\beta|\leq |\alpha|-1$.

Let us list some basic properties of operators of essential order $q$.
\begin{lemma}\label{lemmaR}Let $Q:\cD'(\R^n)\to\cD'(\R^n)$ be an operator of essential order $q$. Then
\begin{enumerate}\setlength{\itemsep}{-1.5pt}
\item\label{lemmaR1} $Q$ is continuous in the $\cD'(\R^n)$ topology;
\item\label{lemmaR2} $Q$ maps $\cD'(\dot\R^n)$ and $\cD'(\{0\})$ to themselves.
\item\label{lemmaR3} Let $\cK_1$ be a linear subspace of $\Do$. Then $Q^{\rm t}$ maps $\cD_{\cK_1}(\R^n)$ to $\cD_{\cK_2}(\R^n)$, where
\[
\cK_2=\{ v\in \Do | \ Q v\in\cK_1\}. 
\]
In particular, $Q^{\rm t}$ maps $\cD(\R^n)$ to $\cD_{\cK}(\R^n)$, where $\cK=\ker (Q|_{\Do})$.
\end{enumerate}    
\end{lemma}
\proof\begin{enumerate}\setlength{\itemsep}{-1.5pt}\item By definition of the weak topology.
\item The first assertion is obvious. For the second one it suffices to notice that for any $v\in\cD'(\{0\})$ the expression $\bra Q v, \varphi \ket$ only depends on the restriction of $\varphi$ to an arbitrary small neighbourhood of $0$.  
\item Let $\varphi\in\cD_{\cK_1}(\R^n)$, then $Q^{\rm t}\varphi \in \cD_{\cK_2}(\R^n)$, since 
for any $v \in \cK_2$, i.e. $v\in\Do$ such that $Qv \in\cK_1$, we have $\bra v, Q^{\rm t}\varphi \ket= \bra Q v ,\varphi \ket=0$.
\end{enumerate}
\qed

The following two lemmas give  examples of operators of essential degree 0, for which we usually reserve the symbol $R$.

\begin{lemma} Let $R$ be an infinitesimal generator of a Lie group $G$ acting  on $\R^n$ such that $0$ is a fixed point. Then $R$ is an operator of essential degree $0$. 
\end{lemma}
\proof Indeed, $R=\sum_{i=1}^n \xi^i(x)\partial_i$ where $\xi^i(0)=0$ (as follows from, i.e., \cite[Ex. 2.68]{olver}). \qed\\

\begin{lemma}Let $\Phi:\R^n\to\R^n$ be a $\cf$ diffeomorphism s.t. $\Phi(0)=0$. Then the operator given by $R u:= u-\Phi^* u$ for $u\in\cD'(\R^n)$ is an operator of essential degree $0$.
\end{lemma}
\proof Since $\Phi(0)=0$, we find that $\supp \varphi\cap\{0\}=\emptyset$ implies $\supp (\Phi^*\varphi)\cap\{0\}=\emptyset$. Let us now check that $\sd (\Phi^* u)\leq \sd\, u$ for all $u\in\cD'(\R^n)$. Indeed, $\tau^{-\omega}\bra u,\varphi(\tau x)\ket\to 0$ for all $\varphi\in\cD(\R^n)$ implies that $\tau^{-\omega}\bra u,(\Phi_*\phi)(\tau x)\ket\to 0$ for all $\phi\in\cD(\R^n)$, so $\tau^{-\omega}\bra \Phi^*u,\phi(\tau x)\ket\to 0$. \qed \\

These two cases are of particular importance in our applications. To see this, first recall that a distribution $u\in\cD'(\R^n)$ is invariant under the induced action of a connected Lie group $G$ acting on $\R^n$ if and only if $R^i u=0$ for all the infinitesimal generators $R^i$ of $G$.  Now, if $0$ is a fixed point of the action of $G$ on $\R^n$, then by the first of the above lemmas, the infinitesimal generators are of essential order $0$.

Similarly, discrete symmetries entail operators of essential degree 0, as they are of the form discussed  in the second lemma. For instance, even distributions are in the kernel of the operator $R^{+}u:=u-u(- \, \cdot)$, and odd ones in that of $R^{-}u:=u+u(- \, \cdot)$.

\subsection{Spaces of distributions supported at the origin}\label{sec:spaces}

Recall that $\Dor$ denotes the finite dimensional vector space spanned by derivatives of the delta distribution up to order $r$. It it will turn out to be very useful to equip it with a scalar product. To this end, for $r\geq 0$ define the maps
\beqa
\cS_r : \Dor \to  \cf(\R^n), \quad \cS_r v:=\sum_{|\alpha|\leq r}\frac{x^{\alpha}}{\alpha !}\bra v, x^{\alpha}\ket, \ v\in\Dor\\
\cT_r :  \cf(\R^n) \to \Dor, \quad \cT_r f:=\sum_{|\alpha|\leq r}\frac{\delta^{(\alpha)}}{\alpha !}\bra\delta^{(\alpha)}, f\ket, \ f\in\cf(\R^n).
\eeqa 
One can easily check  that $\cT_r\cS_r={\rm id}$ on $\Dor$ and $\cS_r\cT_r={\rm id}$ on the space of polynomials of degree $\leq r$. Now set
\begin{eqnarray*}
(v|w)_r:=\bra \bar{v}, \cS_r w\ket =\textstyle \sum_{|\alpha|\leq r}\frac{1}{\alpha !}\bra \bar{v},x^{\alpha}\ket\bra w, x^{\alpha}\ket=\bra w, \cS_r \bar{v}\ket, \quad v,w\in\Dor.
\end{eqnarray*}
where the bar denotes ordinary complex conjugation. Writing elements $v, w$ of $\Dor$ as linear combinations $v=\sum_{|\alpha|\leq r}v_{\alpha}\delta^{(\alpha)}$, $w=\sum_{|\alpha|\leq r}w_{\alpha}\delta^{(\alpha)}$, with $v_{\alpha},w_{\alpha}\in\C$, we have 
\[
\textstyle (v|w)_r=\bra  \sum_{|\alpha|\leq r} \overline{v}_{\alpha}\delta^{(\alpha)},\sum_{|\beta|\leq r} (-1)^{\beta} w_{\beta} x^{\beta}\ket= \sum_{|\alpha|\leq r}\alpha!\, \overline{v}_{\alpha}w_{\alpha},
\]
therefore it is evident that $(\cdot | \cdot )_r$ is a scalar product on $\Dor$. 

Let $Q:\cD'(\R^n)\to\cD'(\R^n)$ be an operator of essential order $q$. We denote
\[
Q|_r : \Dor \to \Do_{\leq r+q}
\]
the restriction of $Q$ to $\Dor$, understood as an operator from $\Dor$ to $\Do_{\leq r+q}$. Let us stress that this definition depends on the essential order of $Q$. The next lemma characterizes the adjoint of $Q|_r$. Its proof relies essentially on the fact that, by assumption, $Q^{\rm t}$ maps $\cf(\R^n)$ to $\cf(\R^n)$.

\begin{lemma}\label{lem:adjoint}Let $Q$ have essential order $q$. Then the adjoint of $Q|_r : \Dor \to \Do_{\leq r+q}$ is 
\beq\label{eq:adjoint}
(Q|_r)^*=\cT_{r} \overline{Q}^{\rm t} \cS_{r+q} :  \Do_{\leq r+q} \to \Dor .
\eeq
Moreover, if $Q^{\rm t}$ maps polynomials of order $\leq r+q$ to elements of
\beq\label{eq:rspace}
\{ f\in C^{\infty}(\R^n): \ f^{(\alpha)}(0)=0, \ |\alpha|>r \},  
\eeq
then for all $s\geq r$ the operator $(Q|_{s})^*$ restricted to $\Do_{\leq r+q}$ is equal to  $(Q|_{r})^*$.
\end{lemma}
\proof For any $v\in\Do_{\leq r+q}$ and $w\in\Do_{\leq r}$ one has
\[
\begin{aligned}
(v|Q w)_{r+q}&=\bra \bar{v},\cS_{r+q} Q w \ket = \bra Qw, \cS_{r+q} \bar{v} \ket \\ & = \bra w, Q^{\rm t} \cS_{r+q} \bar{v}\ket=\bra \cT_s Q^{\rm t}\cS_{r+q} \bar{v},\cS_s w \ket = (\cT_s \overline{Q}^{\rm t} \cS_{r+q} v|w)_{s},
\end{aligned}
\]
so (\ref{eq:adjoint}) follows from the particular case $r=s$. For the second claim, observe that $\cS_{s+q}$ coincides with $\cS_{r+q}$ on $\Do_{\leq r+q}$ and $\cT_s$ coincides with $\cT_r$ on the space (\ref{eq:rspace}). Therefore, $\cT_{s} \overline{Q}^{\rm t} \cS_{s+q}$ restricted to $\Do_{\leq r+q}$ equals $\cT_{r} \overline{Q}^{\rm t} \cS_{r+q}$.
\qed

\subsection{On-shell extension --- single operator case}\label{sec:main2}

Let us specify the problem we already outlined in the introduction, first for the case when only one operator is considered.

\begin{myindentpar}{0.4cm} \emph{Problem} --- Let $Q:\cD'(\R^n)\to\cD'(\R^n)$ be an operator of essential order $q$. Let $u\in\cD'(\dot \R^n)$ have degree of divergence $r:=\deg u<\infty$ and assume
\[
Q u =0 \quad \mbox{ on \ } \dot\R^n.
\]
Find $\ddot u\in\cD'(\R^n)$ such that $\ddot u = u$ on $\dot\R^n$ and $Q\ddot u=0$ on $\R^n$. If such extension(s) exist, we call them \emph{on-shell extensions} w.r.t. $Q$.

\end{myindentpar}

The next theorem is the core of our paper. It is based on the observation that if $u\in\cD'(\dot\R^n)$ is a solution for $Q$ on $\dot\R^n$ and on-shell extensions exist, then an arbitrary extension of $u$ with the same degree of divergence is on-shell modulo an element of $\ran (Q|_r)$ (where $r=\deg u$). If the problem was purely finite-dimensional, we could get rid of the remainder in $\ran (Q|_r)$ by using an orthogonal projection to $(\Ran (Q|_r))^{\bot}$. Observe that any such projection can be expressed as a polynomial in the operator $(Q|_r)^*(Q|_r)$. As we show below, it suffices to consider such polynomials in $(Q|_r)^*Q$ instead, to get an operator with  the desired properties which is well-defined on distributions $\dot u\in\cD'(\R^n)$.

\begin{theorem}\label{thm:core} Let $Q:\cD'(\R^n)\to\cD'(\R^n)$ be an operator of essential order $q$. Let $u\in\cD'(\dot \R^n)$ have $r:=\deg u<\infty$ and assume
\[
Q u =0 \quad \mbox{ on \ } \dot\R^n.
\]
Then the  following statements are equivalent:
\begin{enumerate}    \renewcommand*\theenumi{\alph{enumi}}\renewcommand*\labelenumi{\theenumi)}
\item\label{ita} There is an extension $\ddot u\in\cD'(\R^n)$ of $u$ with $\deg\ddot u=r$ such that $Q \ddot u=0$ on $\R^n$ (that is, $u$ has on-shell extensions);
\item\label{itc} $Q\dot u\in \ran (Q|_r)$ for all extensions $\dot u\in\cD'(\R^n)$ of $u$ with $\deg\dot u=r$;
\item\label{itb} $Q\dot u\in \ran (Q|_r)$ for some extension $\dot u\in\cD'(\R^n)$ of $u$ with $\deg\dot u=r$;
\item\label{itf} $Q p_r\big((Q|_r)^* Q\big)\dot u =0$ for all extensions $\dot u\in\cD'(\R^n)$ of $u$ with $\deg\dot u=r$, where $p_r$ is the polynomial  $p_r(z):=\prod_{\lambda}(1-z/\lambda)$, the product being taken over all nonzero $\lambda\in\sp \big( (Q|_r) (Q|_r)^*\big)$;
\item\label{itd} $Q p_r\big((Q|_r)^* Q\big)\dot u =0$ for some extension $\dot u\in\cD'(\R^n)$ of $u$ with $\deg\dot u=r$, where $p_r$ is as above.

\end{enumerate}
\end{theorem}
\proof \ref{ita})$\Rightarrow$\ref{itc}): Assume $Q \ddot u=0$. Since for all extensions $\dot u$ of $u$ with $\deg\dot u=\deg u=r$ we have $\dot u=\ddot u +v$ for some $v\in\Dor$, it follows that $Q \dot u = Q (\ddot u + v)= Q v\in\ran (Q|_r)$.\\
\ref{itc})$\Rightarrow$\ref{itb}) is obvious.\\
\ref{itb})$\Rightarrow$\ref{ita}): Assume  $Q\dot u= Q v$ for some $v\in\Dor$. Then $\ddot u:= \dot u - v$ satisfies $Q\ddot u =0$.\\
\ref{itc})$\Rightarrow$\ref{itf}): Let us note first that in the expression $p_r\big((Q|_r)^* Q\big)\dot u$, the operator $(Q|_r)^*$ appears always in front of $Q\dot u\in\Do_{\leq r+q}$, so that $(Q|_r)^*$ acts indeed on elements of $Q\dot u\in\Do_{\leq r+q}$, so $p_r\big((Q|_r)^* Q\big)\dot u$ is well-defined. We then have
\[
 Q p_r\big((Q|_r)^* Q\big)\dot u = p_r\big( Q (Q|_r)^*\big) Q \dot u= p_r\big( (Q|_r) (Q|_r)^*\big) Q \dot u
\]
If $Q \dot u \in \Ran (Q|_r)$, then this last expression vanishes because $p_r\big( (Q|_r) (Q|_r)^*\big)$ is the orthogonal projection to $\Ker (Q|_r)^*=(\Ran (Q|_r))^{\bot}$.\\
\ref{itf})$\Rightarrow$\ref{itd}) is obvious.\\
\ref{itd})$\Rightarrow$\ref{ita}): Set $\ddot u := p_r\big((Q|_r)^* Q\big)\dot u$. Since $p_r(0)=1$ and $(Q|_r)^* Q \dot u \in \Dor$, $\ddot u$ is an extension of $u$ and $\deg\ddot u =r$. Therefore $\ddot u$ is a solution with the required properties.\\
\qed\\

Theorem \ref{thm:core} has an important consequence: Given an arbitrary extension $\dot u$, by statement~\ref{itf}),  the distribution
\[
\ddot u := p_r\big((Q|_r)^* Q\big)\dot u
\] 
is a candidate for an on-shell extension. To calculate it explicitly, one only has to find the eigenvalues of the finite-dimensional  matrix $(Q|_r) (Q|_r)^*$.

There is another remarkable feature of the map $\dot u\mapsto p_r\big((Q|_r)^* Q\big)\dot u $. It turns out that under additional conditions (often satisfied in practice), it is linear on the space of extensions to $\R^n$ of solutions of $Q$ on $\dot\R^n$, provided that on-shell extensions exist. Therefore, one can obtain on-shell extensions from generic extensions by applying a certain fixed linear map. This statement is obvious when one speaks only of distributions with a fixed scaling degree $r$. However, it is not at all evident if one considers distributions with different degrees of divergence, for both the definition of $p_r$ and $(Q|_r)^*$ depend on $r$. 

\begin{proposition}\label{prop:linear}
Let $Q$ have essential order $q$ and assume that for all $r\in\N_0$, $Q^{\rm t}$ maps polynomials of order $\leq r+q$ to elements of
\[
\{ f\in C^{\infty}(\R^n): \ f^{(\alpha)}(0)=0, \ |\alpha|>r \}. 
\]
Let $\mathfrak{ V}_Q(\dot\R^n)$ be the space of all $u\in\cD'(\dot\R^n)$ with finite degree of divergence s.t. $Q u=0$ on $\dot\R^n$ and $u$ has an on-shell extension for $Q$. Let $\mathfrak{ V}_Q(\R^n)$ consist of all extensions of elements of $\mathfrak{ V}_Q(\dot\R^n)$ with no greater degree of divergence. The map
\[
\mathfrak{ V}_Q(\R^n)\ni \dot u \mapsto  p_r\big((Q|_r)^* Q\big) \dot u \in \cD'(\R^n)
\]
is linear (where $r=\deg \dot u$ and $p_r$ is as in Theorem \ref{thm:core}).
\end{proposition}
\begin{proof}Let $\dot u_1, \dot u_2\in \mathfrak{ V}_Q(\R^n)$ and denote $r_1=\deg u_1$, $r_2=\deg u_2$ and $r=\deg(u_1+u_2)\leq\max\{r_1,r_2\}$. To prove that
\[
p_r\big((Q|_r)^* Q\big) (\dot u_1 + \dot u_2) =p_{r_1}\big((Q|_r)^* Q\big)\dot u_1 +  p_{r_2}\big((Q|_r)^* Q\big) \dot u_2
\]
it suffices to show that
\[
0=\left[ p_{r'}\big((Q|_{r'})^* Q\big)- p_r\big((Q|_r)^* Q\big)\right ]\dot u
\]
for all $r'\geq r$ and $\dot u \in  \mathfrak{ V}(Q)$ with $\deg \dot u =r $. One has
\[
\left[ p_{r'}\big((Q|_{r'})^* Q\big)- p_r\big((Q|_r)^* Q\big)\right ]\dot u=(p_{r'}-p_r)\big((Q|_r)^* Q\big)\dot u 
\]
because $(Q|_{r'})^*$ restricted to $\Do_{\leq r+q}$ equals $(Q|_r)^*$ by Lemma \ref{lem:adjoint}. Since $(p_{r'}-p_r)(0)=0$, the expression $(p_{r'}-p_r)\big((Q|_r)^* Q\big)$ is the sum of elements of the form $\big((Q|_r)^* Q\big)^k$. For $Q \dot u\in\Do_{\leq r+q}$, we conclude $(p_{r'}-p_r)\big((Q|_r)^* Q\big)\dot u\in\Ran (Q|_r)^*$. On the other hand, $p_{r}\big((Q|_{r})^* Q\big)\dot u\in \ker (Q|_r)$ by Theorem \ref{thm:core} and one can repeat the arguments in the proof therein to show also $p_{r'}\big((Q|_{r'})^* Q\big)\dot u\in \ker (Q|_r)$. It follows that $(p_{r'}-p_r)\big((Q|_r)^* Q\big)\dot u$ belongs to $\ker (Q|_r)=(\Ran (Q|_r)^*)^{\bot}$, hence vanishes.
\end{proof}

\begin{remark} The problem of finding on-shell extensions can be thought as a variant of the following \textit{Bochner's extension problem}:
 \begin{myindentpar}{0.4cm}
 Let $u\in\cD'(\R^n)$ and assume $P u=0$ on $\dot\R^n$. Does $Pu=0$ hold on $\R^n$?
 \end{myindentpar}
If $Pu=0$ on $\R^n$, one says that $u$ has a \emph{removable singularity} for $P$ at $0$, see \cite{ruzhansky} and references therein for a collection of results on that subject. Observe that the assumption $P u=0$ on $\dot\R^n$ implies that $Pu$ is supported in $0$, and a computation of the degree of divergence gives $P u\in\cD'(\{0\})_{\leq r+m}$, where $r=\deg u$ and $m$ is the essential order of $P$. In particular, we obtain that $Pu=0$ on $\R^n$ if $\deg u<-m$. This gives a useful sufficient condition for removable singularities, which can be rephrased in terms of commonly used function spaces such as $L^p$ or Sobolev spaces, more suitable for various applications outside quantum field theory. On the other hand, the theorems proved in our paper are particularly useful in the more singular case when $\deg u\geq -m$ as a distribution on $\dot\R^n$. 
\end{remark}

\subsection{Operators of essential order 0}\label{sec:main0}

As we have seen, operators of essential order $0$ are of special interest in the applications. Moreover, they map $\Dor$ to itself, so one can study the natural subclasses consisting of self-adjoint or normal operators and use their special properties in the analysis. Also, for an operator of essential order $0$, we will give sufficient conditions for the existence of on-shell extensions which are easy to check.

In what follows, $R:\cD'(\R^n)\to\cD'(\R^n)$ is always an operator of essential order $0$.\\

If $R|_r$ is normal, then $\ker(R|_r)=\ker(R|_r)^*=\ker(R|_r)^k$ for $k\in\N_0$. This fact is used in the proof of the next proposition, which provides additional information when no on-shell extension exists for $R$ (or more generally $R^k$). This will for instance be the case for homogeneous distributions of degree $a$ s.t. $-a\in\N_0+n$, which in general do not have homogeneous extensions to $\R^n$.

\begin{proposition}\label{prop:normal}Let $R$ be of essential order $0$. Let $u\in\cD'(\dot\R^n)$ have degree of divergence $r<\infty$ and suppose
\[
 R^k u=0 \mbox{ \ \ on \ } \dot\R^n
\] 
for some $k\in\N_0$. If $R|_r$ is normal, there exists an extension $\ddot u\in\cD'(\R^n)$ of $u$ with $\deg\ddot u=\deg u$ such that 
\[
R^{k+1} \ddot u=0  \mbox{ \ \ on \ } \R^n.
\]
More precisely, one can take $\ddot u = p_r\big((R^k|_r)^* R^k\big)\dot u$, where $\dot u\in\cD'(\R^n)$ is an arbitrary extension of $u$ with $\deg\dot u=\deg u$ and $p_r(z)=\prod_{\lambda}(1-z/\lambda)$, the product being taken over all nonzero $\lambda\in\sp \big( (R^k|_r) (R^k|_r)^*\big)$.
\end{proposition}
\proof One has 
\[
R^{k+1}\ddot u =R^{k+1} p_r\big((R^k|_r)^* R^k\big)\dot u= R p_r \big( R^k (R^k|_r)^* \big) R^k \dot u = R p_r \big( (R^k|_r) (R^k|_r)^* \big) R^k \dot u.
\]
 This vanishes because $R^k\dot u\in\Dor$ and $p_r \big( (R^k|_r) (R^k|_r)^* \big)$ is the orthogonal projection to $\ker (R^k|_r)^*$, which by normality of $R|_r$ equals $\ker (R|_r)$.\qed\\

An analogue of Proposition \ref{prop:linear} is available:

\begin{proposition}\label{prop:normal2}Let $R$ be of essential order $0$ and assume that for each $r\in\N_0$, $R|_r$ is normal and $R^{\rm t}$ maps polynomials of order $\leq r$ to elements of 
\[
\{ f\in C^{\infty}(\R^n): \ f^{(\alpha)}(0)=0, \ |\alpha|>r \}. 
\]
Let $\mathfrak{N}_R(\dot\R^n)$ be the space of all $u\in\cD'(\dot\R^n)$ with finite degree of divergence such that $R^k u=0$ on $\dot\R^n$ for some $k\in\N_0$. Let $\mathfrak{N}_R(\R^n)$ consist of all extensions of elements of $\mathfrak{N}_R(\dot\R^n)$ with no greater degree of divergence. The map
\[
\mathfrak{N}_R(\R^n)\ni \dot u \mapsto  p_r\big((R^k|_r)^* R^k\big) \dot u \in \cD'(\R^n)
\]
is linear, where $p_r$ is as in Proposition \ref{prop:normal} and $k$ is taken to be sufficiently high.
\end{proposition}
\proof Let $r=\deg\dot u$ and let $k\in\N_0$ be such that $R^k u=0$. Analogously to the proof of Proposition \ref{prop:normal}, one shows that for any $\dot u\in\mathfrak{N}_R(\R^n)$, $r'\geq r$ and $k'\geq k$,
\[
\left[p_{r'}\big((R^{k'}|_{r'})^* R^{k'}\big)-p_r\big((R^k|_r)^* R^k\big)\right]\dot u\in \Ran(R^k)^*
\]
and that this expression also belongs to $\ker(R|_r)=\ker(R|_r)^k=(\ran(R^k|_r)^*)^{\bot}$, hence vanishes.\qed\\

In the case when on-shell extensions for $R$ exist, they can also be obtained as follows using the resolvent of $R|_r$.

\begin{proposition}Let $R$ be an operator of essential order $0$. Assume $u\in\cD'(\dot\R^n)$ has degree of divergence $r<\infty$, satisfies $Ru=0$ and has on-shell extensions. If $R|_r$ is normal then
\[
\ddot u :=\lim_{\varepsilon\to 0}\big(1-\left[(R-\i\varepsilon)|_r\right]^{-1}R\big) \dot u
\]
is an on-shell extension, where $\dot u$ is an arbitrary extension of $u$ to $\R^n$ with the same degree of divergence.
\end{proposition}
\proof By continuity of $R$,
\begin{eqnarray*}
R\ddot u=R\lim_{\varepsilon\to 0}\big(1-\left[(R-\i\varepsilon)|_r\right]^{-1}R\big) \dot u=\lim_{\varepsilon\to 0}\big(1-R\left[(R-\i\varepsilon)|_r\right]^{-1}\big)R \dot u\\=\lim_{\varepsilon\to 0}\big(1-(R|_r)\left[(R-\i\varepsilon\right)|_r]^{-1}\big)R\dot u=\lim_{\varepsilon\to 0}(-\i\varepsilon)\left[(R-\i\varepsilon\right)|_r]^{-1}R\dot u.
\end{eqnarray*}
The operator $(-\i\varepsilon)\left[(R-\i\varepsilon\right)|_r]^{-1}$ converges to the orthogonal projection to $\Ker R|_r=\Ker (R|_r)^*$ and by assumption $R\dot u\in\Ran R|_r$, therefore the above expression vanishes.
\qed\\

Let us now examine the special case when $R|_r$ is self-adjoint. Then, the operator $p_r\big((R|_r)^*R\big)$ in Theorem \ref{thm:core} can be replaced by $p_r\big( R^2 \big)$. The gain is that $p_r\big(R^2 \big)$ makes sense as an operator acting on arbitrary distributions. Even better, it is well defined as an operator from $\cD'_{\cK}(\R^n)$ to $\cD'(\R^n)$, where $\cK=\Ker(R|_r)$. This fact can be used to construct \emph{directly} an on-shell extension $\ddot u$, that is without referring to some generic extension $\dot u$. In the following proposition, the  polynomial $p_r$  (defined using the eigenvalues of $R^2$) is replaced by a polynomial $b_r$ defined using the eigenvalues of $R$, which makes the formulae slightly more compact.
 
\begin{proposition}\label{prop:direct}Let $R$ be of essential order $0$ and assume $(R|_r)^*=R|_r$. Assume $u\in\cD'(\dot\R^n)$ has degree of divergence $r<\infty$ and satisfies $Ru=0$ on $\dot\R^n$. Suppose that $u$ has an on-shell extension for $R$. Denote $\cK=  \Ker(R|_r)$ and let $W^{\rm t}:\cD(\R^n)\to \cD_{\cK}(\R^n)$ be a projection. Set 
\beq
\ddot u := W b_r(R) \tilde u,
\eeq
where $\tilde u$ is the unique extension of $u$ in $\cD'_r(\R^n)$ with $\deg \tilde u =\deg u$, and $b_r$ is the polynomial  
\[
b_r(z)=\textstyle\prod_{\lambda\in\sp R|_r\setminus\{0\}}(1-z/\lambda).
\]
Then $\ddot u$  is an extension of $u$ in $\cD'(\R^n)$ with $\deg \ddot u =\deg u $ and $R \ddot u=0$ on $\R^n$.
\end{proposition}
\proof The projection $W^{\rm t}$ can be written as
\[
W^{\rm t}\varphi=\varphi-\sum_{i\in \cI}\bra w_i,\varphi \ket \phi_i,
\]
where $\{w_i\}_{i\in\cI}$ is a basis of $\cK=\ker (R|_{r})=\ran (b_r(R|_r))$ and $\{\phi_i\}_{i\in\cI}$ are elements of $\cD(\R^n)$ such that $\bra w_i,\phi_j\ket=\delta_{ij}$. Let us choose $\{v_i\}_{i\in\cI}$, $v_i\in\cK$ in such way that $b_r(R)v_i=w_i$ (which is always possible for dimensional reasons). Set
\[
V^{\rm t}\varphi:=\varphi-\sum_{i\in \cI}\bra v_i,\varphi \ket b_r(R^{\rm t})\phi_i
\]
Since $\bra v_i, b_r(R^{\rm t})\phi_j\ket=\bra b_r(R) v_i, \phi_j\ket=\delta_{ij}$, $V^{\rm t}:\cD(\R^n)\to\cD_{\cK^{\bot}}(\R^n)$ is a projection and a short computation gives $V^{\rm t}b_r(R^{\rm t})=b_r(R^{\rm t})W^{\rm t}$. Since $b_r(R^{\rm t})$ maps $\cD(\R^n)$ to $\cD_{\cK^{\bot}}(\R^n)$, we have $Z^{\rm t}b_r(R^{\rm t})=b_r(R^{\rm t})$ for any projection $Z^{\rm t}:\cD(\R^n)\to\cD_{\cK^{\bot}}(\R^n)$. Thus,
\[
\ddot u = W b_r(R)\tilde u = b_r(R) Z V \tilde u.
\]
By \ref{lemmaR3}. of Lemma \ref{lemmaR}, $V^{\rm t}$ maps $\cD_{\cK^{\bot}}(\R^n)$ to $\cD_r(\R^n)$, hence $V^{\rm t}Z^{\rm t}$ maps $\cD(\R^n)$ to $\cD_r(\R^n)$. It follows that $\dot u:=Z V \tilde u$ is an element of $\cD'(\R^n)$ and $\ddot u=b_r(R)\dot u$. By Theorem \ref{thm:core} (with $b_r(R)$ playing the role of $p_r(R^2)$), $\ddot u$ is an element of $\cD'(\R^n)$ with the required properties.
\qed\\

In the next theorem we give several conditions on $R$ which ensure the existence of on-shell extensions for all degrees of divergence $r\geq 0$. Note that these conditions are of rather different nature.

\begin{theorem}\label{thm:criterion}Let $R$ be an operator of essential order $0$, let $u\in\cD'(\dot\R^n)$ have degree of divergence $r<\infty$ and assume $Ru=0$ on $\dot\R^n$. Assume at least one of the following holds:
\begin{enumerate}    \renewcommand*\labelenumi{\theenumi)}
\item\label{ipa} $\ker (R|_r)=\{0\}$;
\item\label{ipb} $R^{\rm t}$ maps polynomials of degree $\leq r$ to polynomials of degree $\leq r$ and $\supp\, u$ is compact;
 \item\label{ipc} $R^{\rm t}$ maps polynomials of degree $\leq r$ to polynomials of degree $\leq r$ and there exist $\psi,\phi\in\cD(\R^n)$ s.t. $\psi\equiv 1$ in a neighbourhood of $0$ and $R^{\rm t}\psi=\phi R^{\rm t}$;
\end{enumerate}
Then $u$ admits an on-shell extension, i.e. an extension $\ddot u\in\cD'(\R^n)$ s.t. $\deg\ddot u = r$ and  $R\ddot u=0$ on $\R^n$.\end{theorem}
\proof \ref{ipa}) If $\ker (R|_r)=\{0\}$ then $\ran (R|_r)=\Dor$ and \ref{itc}) of Theorem \ref{thm:core} is trivially satisfied.

\ref{ipb}) Let $\dot u\in\cD'(\R^n)$ be an arbitrary extension of $u$ with the same degree of divergence. We want to show  that \ref{itc}) of Theorem \ref{thm:core} is satisfied, or equivalently
\beq\label{eq:toshow}
(R\dot u|v)_r=0 \quad \forall\,v\in(\Ran (R|_r))^{\bot}.
\eeq
Using that $u$ is compactly supported, we obtain
\begin{eqnarray*}
(R\dot u|v)_r=\bra \overline{ R \dot u},\cS_r v \ket=\bra \overline{  \dot u},\overline{R^{\rm t}}\cS_r v \ket.
\end{eqnarray*}
To show that the expression above vanishes, let us remark that $v\in\Ran (R|_r))^{\bot}$ means 
\[
0=(Rw| v)_r = \bra \overline{R w}, \cS_r  v \ket=\bra\overline{w}, \overline{R^{\rm t}} \cS_r  v \ket \quad \forall w\in\Dor,
\]
which  implies $\overline{R^{\rm t}} \cS_r  v=0$.

\ref{ipc}) Let us show that (\ref{eq:toshow}) holds. We have 
\begin{eqnarray*}
(R\dot u|v)_r=\bra \overline{ R \dot u},\cS_r v \ket=\bra \overline{ R \dot u},\overline{\psi}\cS_r v \ket=\bra\overline{\dot u},\overline{R}^{\rm t} \overline{\psi}\cS_r v\ket
=\bra\overline{\dot u},\overline{\phi}\overline{R}^{\rm t} \cS_r v\ket,
\end{eqnarray*}
where in the second equality we used that $R \dot u$ is supported at $\{0\}$. The expression above vanishes because as previously, $\overline{R}^{\rm t} \cS_r v=0$.
\qed\\

Let us emphasize that if condition \ref{ipa}) holds then the on-shell extension $\du$ is unique.

Condition \ref{ipc}) is satisfied if $R$ is for instance one of the infinitesimal generators of rotations.

\subsubsection{Example --- homogeneous and almost homogeneous distributions}\label{sec:homogeneous}

The canonical example for extension of singular distributions are homogeneous distributions. We will now show that the known results on extensions of homogeneous or almost homogeneous distributions which appear in renormalisation are easily recovered in our approach .

\begin{proposition} (\cite{hoermander1}, Thm 3.2.3) \label{prophext}Let $u\in\cD'(\dot\R^n)$ be homogeneous of degree $a\in\C$, i.e.
\[
(\textstyle\sum_{i=1}^n x_i\partial_i - a) u = 0 \ \ \ \mbox{ on \ } \dot\R^n.
\]
If $-\Re a\notin\N_0+n$ then $u$ has a unique homogeneous extension $\dot u\in\cD'(\R^n)$ (i.e. $(\sum_{i=1}^n x_i\partial_i - a) \dot u=0$ on $\R^n$).
\end{proposition} \proof Clearly, $R:=\sum_{i=1}^n x_i\partial_i - a$ is an operator of essential order $0$. By Theorem \ref{thm:criterion}, for an on-shell extension to exist it is sufficient that $\ker (R|_{r})=\{0\}$  for all $r\in\N_0$ (and in such case it is unique). Since $R|_{r}\delta^{(\alpha)}=R\delta^{(\alpha)}=-(|\alpha|+n+a)\delta^{(\alpha)}$, we have
\[
\left|\det R|_{r}\right| = \left| \textstyle\prod_{|\alpha|\leq -\Re a-n } (|\alpha|+n+a)\right|,
\]
so that $-\Re a\notin\N_0+n$ entails $\left|\det R|_{r}\right|\neq 0$.
\qed\\

The following proposition concerns distributions of generalized homogeneity. It is a variant of \cite[Lem. 4.1]{hollwald}, see also \cite[Prop. 4]{DF04}) and \cite[Cor. I.15]{keller}.

\begin{proposition} Let $a_1, \dots, a_k \in \C$. Let $u\in\cD'(\dot\R^n)$ and assume $\prod_{j=1}^k R(a_j)u = 0$ on $\dot\R^n$, where $R(a):=(\sum_{i=1}^n x_i\partial_i - a)$. If $-\Re a_j\notin\N_0+n$, then $u$ has a unique extension $\ddot u\in\cD'(\R^n)$ s.t. $\prod_{j=1}^k R(a_j) \ddot u = 0$ on $\R^n$.
 \end{proposition} \proof We use Theorem \ref{thm:criterion} as in Proposition \ref{prophext}. We need to prove that $\prod_{j=1}^k R(a_j)$ restricted to $\Dor$ has trivial kernel for $r\in\N_0$. But this readily follows from the same property for the operators $R(a_j)$ that was shown in the proof of Proposition \ref{prophext}.\\
\qed

If $u$ is as above with all $a_j$'s equal, one speaks of an \emph{almost homogeneous distribution} (or \emph{associate homogeneous distribution}, cf. \cite{nikolovtodorov} and references therein). If the $a_j$'s are pairwise distinct, one speaks of a \emph{heterogeneous distribution}. 

Since for arbitrary $r\in\N_0$, $R(a)|_r$ is diagonal in the basis $\{\delta^{(\alpha)}\}_{|\alpha|\leq r}$, it is normal. As a straightforward corollary from Proposition \ref{prop:normal} one recovers the following result on almost homogeneous distributions.

\begin{proposition} Let $u\in\cD'(\dot\R^n)$ and assume $R(a)^k u = 0$ on $\dot\R^n$. Then there exists an extension $\dot u\in\cD'(\R^n)$ with the same degree of divergence s.t. $R(a)^{k+1}\dot u=0$ on $\R^n$. 
\end{proposition}

\subsection{On-shell extension --- multiple operators}

Let us now move on to the more general problem where instead of a single operator $Q$,  several operators $\{ Q^i\}_{i=1}^{k}$ are considered.

If the operators $\{ Q^i\}_{i=1}^{k}$ commute pairwise, one can easily generalize the results of sections \ref{sec:main2} and \ref{sec:main0}. For instance, the generalization of the key part of Theorem \ref{thm:core} reads:   

\begin{theorem} Let $\{Q^i\}_{i=1}^k$ be a family of mutually commuting operators of arbitrary essential order. Let $u\in\cD'(\dot \R^n)$ have $r:=\deg u<\infty$ and let it satisfy
\[
Q^i u =0 \quad {\mbox{on} \ } \dot\R^n, \quad i=1,\dots, k.
\]
The following are equivalent:
\begin{enumerate} \renewcommand*\theenumi{\alph{enumi}}\renewcommand*\labelenumi{\theenumi)}
\item\label{itap} There is an extension $\ddot u\in\cD'(\R^n)$ of $u$ with $\deg\ddot u=\deg u$ such that $Q^i \ddot u=0$ on $\R^n$ ($i=1,\dots,k$);
\item\label{ite} For all extensions  $\dot u\in\cD'(\R^n)$ of $u$ with $\deg\dot u=\deg u$, one has 
\[
Q^j\prod_{i=1}^k p^i_r\big((Q^i|_r)^* Q^i\big)\dot u =0, \quad j=1,\dots, k,
\]
where $p^i_r$ is the polynomial  $p^i_r(z):=\prod_{\lambda}(1-z/\lambda)$, the product being taken over all nonzero $\lambda\in\sp \big( (Q^i|_r) (Q^i|_r)^*\big)$.
\end{enumerate}
\end{theorem}

In the case when the operators $Q^i$ do not commute pairwise, one strategy is to find a  polynomial of the $Q^i$'s (or several mutually commuting ones),  which commutes (respectively, commute) with all the $Q^i$'s. If the set of solutions of this operator (respectively, the joint set of solutions of these operators) coincides with the joint set of solutions of the $Q^i$'s, then one is reduced to the case of a single operator (or several mutually commuting ones). This requirement can be formulated as follows:

\begin{myindentpar}{0.4cm}
\emph{Assumption C} --- Assume there exist mutually commuting operators $\{C^j\}_{j=1}^{k'}$ which are polynomials of the $Q^i$'s, commute with all the $Q^i$'s and satisfy
\[
\textstyle\bigcap_{j=1}^{k'}\ker (C^j|_r)= \textstyle\bigcap_{i=1}^k \ker (Q^i|_r).
\]
\end{myindentpar}
Of course, provided mutually commuting operators exist, one inclusion is always guaranteed. The non-trivial part in the assumption is that the joint kernel of the operators $C^j$ should not be larger than that of the original operators. Observe that often in the applications, those $Q^j$ which do not commute among themselves form a Lie algebra, and the $C^j$ are the Lie algebra's Casimir operators. Below we give a criterion for existence of on-shell extensions which is particularly useful in this context.

\begin{theorem}\label{thm:withcasimir} Let $\{R^i\}_{i=1}^k$ be a set of operators of essential degree $0$, let $u\in\cD(\dot \R^n)$ have degree of divergence $r<\infty$ and assume
\[
R^i u=0 \ \ {\rm on} \ \dot\R^n, \quad i=1,\dots,k.
\]
Let $C$ be a polynomial in the variables $R^i$ ($i=1,\dots,k$) with no term of degree one or zero. Assume that
\[
(C|_r)^*=C|_r,\quad C R^i = R^i C,  \ \ i=1,\dots,k.
\]
and that the following stronger form of \emph{Assumption C} is satisfied
\[
\ker C =\textstyle\bigcap_{i=1}^k \ker R^i.
\]
Then $u$ has an on-shell extension, i.e., an extension $\ddot u\in\cD'(\R^n)$ with $\deg u=r$ s.t. $R^i\ddot u=0$ on $\R^n$ for $i=1,\dots,k$.
\end{theorem}
\proof Since $C|_r$ is self-adjoint, there is a polynomial $b_r$ s.t. $b_r(C|_r)$ is the orthogonal projection to $\Ker (C|_r)^*=\Ker \,C|_r$, namely 
\[
b_r(z)=\textstyle\prod_{\lambda\in\sp C|_r\setminus\{0\}}(1-z/\lambda).
\]
 Let $\dot u\in\cD'(\R^n)$ be an arbitrary extension of $u$ with  degree of divergence $r$ and set $\ddot u:=b_r(C)\dot u$. Clearly, $\ddot u$ is an extension of $u$ with degree of divergence $r$. Moreover,
\[
R^i \ddot u = R^i b_r(C)\dot u=b_r(C|_r)R^i \dot u\in \ker\, C=\textstyle\bigcap_{j=1}^k \ker R^j, \quad i=1,\dots,k
\]
hence $R^i R^j \ddot u=0$ for all $i,j$ and consequently $C \ddot u=0$, which entails $R^i\ddot u=0$.
\qed\\

The above theorem can be used to treat Lorentz symmetry, by taking $C=(x_{\mu}\partial_{\nu}-x_{\nu}\partial_{\mu})(x^{\mu}\partial^{\nu}-x^{\nu}\partial^{\mu})$ (the quadratic Casimir for the Lorentz group) and $R_i$ proportional to generators of rotations and boosts. It was proved \cite{DF04} that the kernel of this operator corresponds indeed to Lorentz invariant distributions, the operator $b_r(C)$ was also used therein in the construction of on-shell extensions.

\subsubsection*{Application --- renormalisation conditions in scalar theory}

As an application of our framework, we show how can one treat the symmetries which arise in a scalar quantum field theory. In this particular case our method essentially reduces to the arguments used in \cite{DF04}, except that we obtain an additional result on linearity of the map $\cR$ defined below.

\begin{corollary} Let $\mathfrak{V}(\dot\R^n)$ denote the space of all distributions $u\in\cD'(\dot\R^n)$ with finite degree of divergence such that:
\begin{enumerate}\setlength{\itemsep}{-3pt}
\item[a.] $u$ is Lorentz-invariant;
\item[b.] $u$ is the finite sum of almost homogeneous distributions of integer degree.
\end{enumerate}
\smallskip
Let $\mathfrak{V}(\R^n)$ be the space of extensions of elements of $\mathfrak{V}(\dot\R^n)$ to $\R^n$ with no greater degree of divergence. Then there is a linear map $\cR:\mathfrak{V}(\R^n)\mapsto\mathfrak{V}(\R^n)$ s.t. for all $\dot u\in\mathfrak{V}(\R^n)$
\begin{enumerate}\setlength{\itemsep}{-3pt}
\item $\cR\dot u = \dot u$ on $\dot\R^n$ 
\item\label{it:prop1} $\cR\dot u$ is Lorentz-invariant;
\item\label{it:prop2} $\cR\dot u$ is the finite sum of almost homogeneous distributions. More precisely, if $\dot u$ is almost homogeneous of degree $a\in\Z$ and order $k\in\N_0$ on $\dot\R^n$, then $\cR\dot u$ is almost homogeneous on $\R^n$ of degree $a$ and order $k$ if $k\notin\N_0+n$ and of order $k+1$ otherwise.
\end{enumerate}
\end{corollary}
\proof Let $C=(x_{\mu}\partial_{\nu}-x_{\nu}\partial_{\mu})$ and $R(a)=\sum_{i=1}^n x_i\partial_i-a$ as in subsection \ref{sec:homogeneous}. For any $u\in\mathfrak{V}(\R^n)$, $C \dot u=0$ on $\dot\R^n$ and there is a sequence of non-negative integers $\{N_j\}_{j\in\Z}$ s.t. $N_j=0$ for almost all $j$ and $\prod_{j\in\Z}R(j)^{N_j} \dot u=0$ on $\dot\R^n$. Set 
\[
\cR\dot u= p_r({\textstyle \prod_{j\in\Z}R(j)^{2 N_j}}) b_r(C) \dot u,
\]
where $r=\deg\dot u$, $b_r(z):=\prod_{\lambda}(1-z/\lambda)$ where the product runs over all $\lambda\in\sp(C|_r)\setminus\{0\}$, $p_r(z):=\prod_{\lambda}(1-z/\lambda)$ where the product runs over all $\lambda\in\sp\big(\prod_{j\in\Z}R(j)^{2 N_j}|_r\big)\setminus\{0\}$. Since $C|_r$ and $\prod_{j\in\Z}R(j)^{N_j}|_r$ are self-adjoint for any $r$, the arguments from Proposition \ref{prop:normal2} give linearity of $\cR$. Property \ref{it:prop2} is  proved as in subsection \ref{sec:homogeneous}.
Based on the fact that $C$ satisfies the conditions 
in Theorem~\ref{thm:withcasimir}, we deduce that property \ref{it:prop1} holds.
That these two properties can be satisfied simultaneously is due to the fact that the $R(j)$ commute with $C$. 
\qed

\subsection{Vector-valued distributions}

Let us consider distributions in $\cD'(\dot\R^n,\R^q)$, $\cD'(\R^n,\R^q)$. We use the notation $\bra u, \varphi\ket=\sum_{i=1}^q\bra u_i,\varphi_i\ket$ for the pairing between $\cD'(\R^n,\R^q)$ and $\cD(\R^n,\R^q)$, where $u=(u_1,\dots,u_q)\in\cD'(\R^n,\R^q)$ and $\varphi=(\varphi_1,\dots,\varphi_q)\in\cD'(\R^n,\R^q)$. With this notation, the definition of the scaling degree extends verbatim to the case of vector-valued distributions, and so do the results on extension of distributions. One easily sees that
\[
\sd\,u=\max_{i=1,\dots,q} ( \sd\, u_i).
\]

In practice, one is often interested in the following situation. Let $G$ be a Lie group acting on $\R^n$ and consider the group action on $\R^n\times \R^q$  given by 
\beq\label{eq:defmult}
g: \ (x,u)\mapsto (g\cdot x, \mu(g,x) u), \quad g\in G, \  x\in\R^n, \ u\in\R^q,
\eeq
where $\mu: G\times\R^n \to {\rm{GL}}(q)$ satisfies
\[
\mu(g\cdot h, x)=\mu(g, h\cdot x)\mu(h,x), \quad \mu(e,x)=\one
\]
for all $g,h\in G$, $x\in\R^n$. This property of $\mu$ guarantees that (\ref{eq:defmult}) defines an action of $G$. In applications in renormalisation one is mostly interested in the case $\mu(g,x)$ does not depend on $x$. The associated \emph{infinitesimal generators} are of the form
\[
R=\sum_{i=1}^n \xi^i(x)\frac{\partial}{\partial x^i}+\sum_{\alpha,\beta=1}^q h^{\alpha}_{\beta}(x)u^{\beta}\frac{\partial}{\partial u^{\alpha}},
\]
where $\sum_{i=1}^n \xi^i(x)\frac{\partial}{\partial x^i}$ is the infinitesimal generator of $G$ acting on $\R^n$, associated to some element $\bv$ of the Lie algebra of $G$ and $h(x)=\frac{\d}{\d t}\mu(\exp(\bv t),x)|_{t=0}$ (see \cite{olver} for details). Now, the results of sections \ref{sec:main2} and \ref{sec:main0} directly carry over, and extensions which are on-shell w.r.t. the above operator $R$ can be constructed in the same manner as described there. 


\section{On-shell and off-shell time-ordered products}\label{sec:products}

In this final section, we will rephrase the link between on-shell and off-shell time-ordered products within our framework. We will assume  knowledge of perturbative quantum field theory in position space. First recall that the time-ordered product in scalar quantum field theory (of second order) is formally written  as
\[
T\big(\varphi(x)\varphi(y)\big)=\theta(x^0-y^0)\varphi(x)\varphi(y) + \theta(y^0-x^0)\varphi(y)\varphi(x),
\]
where $\varphi(x)$ is the free field  and $\theta$ is the Heaviside theta distribution. By Wick's theorem, $T\big(\varphi(x)\varphi(y)\big)$ is equal to  the normal product $:\!\varphi(x)\varphi(y)\!:$ plus the singular distribution
\beq\label{eq:feyn}
\theta(z^0)\Delta_+ (z)+\theta(-z^0)\Delta_+ (-z),
\eeq
where $z=x-y$ and $\Delta_+$ is the positive-frequency solution for $\dal+m^2$, and where the product of $\theta$ and $\Delta_+$ is well-defined by H\"ormander's criterion \cite[Thm 8.4]{hoermander1} as a distribution in $\cD'(\R^n)$. It is in fact the Feynman propagator $\Delta_{\rm F}\in\cD'(\R^n)$.
%
%
One can also view (\ref{eq:feyn}) as a distribution on $\dot\R^n$. Since its degree of divergence is $-2$, it admits a unique extension to $\R^n$ with the same degree of divergence as given by Theorem~\ref{thm:core}. Of course, this construction again yields the Feynman propagator on $\R^n$.

Ambiguities arise when one considers higher derivatives of the fields. For instance, applying Wick's theorem to $T\big(\partial_x^{\mu}\varphi(x)\partial_y^{\nu}\varphi(y)\big)$, the contribution which cannot be extended unambiguously, is 
\[
-\theta(z^0)\partial^{\mu}\partial^{\nu}\Delta_+ (z)-\theta(-z^0)\partial^{\mu}\partial^{\nu}\Delta_+ (-z).
\]
The degree of divergence of this distribution is $0$. Consequently, its extensions to $\R^n$ are no longer uniquely fixed by requiring that they have the same degree of divergence. Requiring additionally Lorentz covariance, one obtains that the most general form of any such extension is 
\[
\partial^{\mu}\partial^{\nu}\Delta_{\rm F}(z)+C g^{\mu\nu} \delta(z),
\]
where $C$ is an arbitrary constant. The choice $C=0$, or more generally, a prescription for the time-ordered product which would make it `commute' with derivatives, seems to be the simplest one. Such a choice,  however, is inconsistent with the requirement that 
the fields are on-shell in the sense of the equation of motion, i.e. in our language on-shell w.r.t. the Klein-Gordon operator,  $(\dal+m^2)\varphi=0$. Indeed, setting $C=0$ would imply for instance
\[
(\dal_{x}+m^2) T\big(\varphi(x)\varphi(y)\big)=(\dal_{x}+m^2)\Delta_{\rm F}(x-y)=-\i\delta(x-y),
\]
whereas the Klein-Gordon on-shell condition yields $T\big((\dal_{x}+m^2)\varphi(x)\varphi(y)\big)=0$.

Most of the physics literature uses on-shell time-ordered products. The off-shell formalism, developed in \cite{DF03,DF04}, is based on a time-ordered product which commutes with the derivatives of the fields. For this reason it has many advantages over the on-shell formalism, an especially remarkable one being the possiblity of writing in a compact form the so-called Master Ward Identity, which serves as a universal renormalisation condition. 

\subsection{Definition of the off-shell to on-shell map}

Let us briefly recall the definition of the off-shell and on-shell algebras of fields. An off-shell field $\varphi(x)$ is an evaluation functional on $\cf(\R^n)$ (the classical configuration space), namely $\big(\varphi(x)\big)(h):=h(x)$ for $h\in\cf(\R^n)$. The derivatives of off-shell fields are defined by $\big(\partial^{\alpha}\varphi(x)\big)(h):=\partial^{\alpha}h(x)$. The \emph{algebra of off-shell fields} $\cP$ is the commutative algebra generated by elements of the form $\partial^{\alpha}\varphi(x)$ with respect to the pointwise product, i.e. $\big(\partial^{\alpha_1}\varphi(x)\partial^{\alpha_2}\varphi(x)\big)(h):=\partial^{\alpha_1}h(x)\partial^{\alpha_2}h(x)$. The \emph{algebra of on-shell fields}  is the quotient algebra $\cP_0:=\cP/\cJ$, where $\cJ$ is the ideal
\[
\cJ:=\{\textstyle \sum_{\alpha\in\N_0}B_{\alpha}\partial^{\alpha}(\dal+m^2)\varphi : \ B_{\alpha}\in\cP \}.
\]
 We denote by $\pi:\cP\to\cP_0$ the canonical surjection, i.e. $\pi(B)=B+\cJ$ for $B\in\cP$. One can easily see that it is a homomorphism of algebras. The derivatives of on-shell fields are defined for $A\in\cP_0$ by  $\partial^{\mu}A:=\pi\partial^{\mu} B$, where $B$ is an arbitrary element of $\cP$ such that $\pi(B)=A$. One can check that this does not depend on the choice of $B$ and that one has $(\dal+m^2)\pi\varphi=0$ in $\cP_0$ for any $\varphi\in\cP$ --- this is the on-shell property. 

The on-shell time-ordered product $\Ton$ of order $k$ is a map from $\cP_0^{\otimes k}$ to operator-valued distributions. The axioms defining $\Ton$ and its inductive construction are the basic components of the Epstein and Glaser approach to renormalisation and are a subject covered exhaustively by the literature \cite{epsteinglaser}. The off-shell time-ordered product $\Toff$ of order $k$ is a map from $\cP^{\otimes k}$ to operator-valued distributions which satisfies axioms that are analogous to those defining $\Ton$, see \cite{DF03} for a detailed discussion. In the present setting, we are merely interested in the following result which relates $\Ton$ and $\Toff$ :
\begin{theorem}[\cite{DF03,DF04,onshell}]\label{thm:chion}
There exists a unique linear map $S\mapsto\chi(S)$ from differential operators with constant coefficients to differential operators with constant coefficients such that:
\begin{enumerate}\renewcommand*\theenumi{\alph{enumi}}\renewcommand*\labelenumi{\theenumi)}
\item $\chi(S(\dal+m^2))=0$ for all $S$;
\item for any $S=\partial_{\mu_1}\dots\partial_{\mu_k}$, $\chi(S)$ transforms under Lorentz transformations as $S$;  
\item $\ord \,\chi(S) \leq\ord\, S $ for all $S$;
\item there is a linear map $S\mapsto\chi_1(S)$ s.t. $\chi(S)-S=\chi_1(S)(\dal+m^2)$. 
\end{enumerate}
\medskip Now, on-shell and off-shell time-ordered products of order $k$ are related by
\[
\Ton(A_1\otimes\dots \otimes A_k) =\Toff (\sigma(A_1)\otimes\dots\otimes \sigma(A_k)), \quad A_i\in\cP_0,
\]
where $\sigma:\cP_0\to\cP$ is the unique linear algebra homomorphism s.t. $\sigma\pi(S\varphi)=\chi(S)\varphi$ for all $S$, $\varphi$.
\end{theorem}This recurrence has been solved in \cite{onshell} in an explicit way. Although the result is certainly well-suited for practical use, it involves long combinatorial expressions which do not seem to have a deeper interpretation. For the sake of completeness, we quote below the explicit formula for $\chi$ obtained in \cite{onshell}.

\begin{theorem}[\cite{onshell}]\label{thm:brouduetsch} The map $\chi$ from Theorem \ref{thm:chion} is given by
\[
\chi(\partial_{\mu_1}\dots\partial_{\mu_k})=\sum_{j=0}^{k/2}\alpha_j^k P_j^k ( \partial_{\mu_1}\dots\partial_{\mu_k} ),
\]
where $P_j^k(S)=\frac{1}{j!}\Lambda^j(S)$, \   $\Lambda(\partial_{\mu_1}\dots\partial_{\mu_k} )=\sum_{i<j}g_{\mu_i\mu_j}\partial_{\mu_1\dots \hat{i}\dots\hat{j}\dots\partial_{\mu_k}}$ (where $\hat i$ means $\mu_i$ is removed), \ $a_0^k=0$ and for $j>0$
\[
\textstyle \alpha_j^k=(-1)^j (\Box+m^2)\sum^{j-1}_{p=0}{{j-1}\choose {p}}m^{2p}\Box^{j-1-p}\prod^{j-1}_{q=0}(n+2k-2p-2q-4)^{-1}.
\]
\end{theorem}

\subsection{Construction of the map in the present setting}\label{sec:ourmap}

In our setting, the problem can be formulated as follows. We are given a fundamental solution $\Delta_{\rm F}\in\cD'(\R^n)$ of $Q=\dal+m^2$. Given a partial differential operator $S$ with constant coefficients, we want to replace $S\Delta_{\rm F}$ with a distribution which agrees with the latter on $\dot\R^n$ and is, moreover, $0$ when $S=(\dal+m^2)$. More precisely, the question is to associate to each differential operator $S$ with constant coefficients a distribution $\Theta(S)\in\cD'(\R^n)$ such that:
\begin{enumerate}
\item $\Theta(S)=S\Delta_{\rm F}$ on $\dot\R^n$;
\item $\deg \Theta(S)\leq \deg S\Delta_{\rm F}$;
\item the assignment $S\mapsto\Theta(S)$ is linear;
\item $\Theta(S)$ is Lorentz-covariant; 
\item $\Theta(S(\dal+m^2))=0$ for all $S$.
\end{enumerate}
\medskip Provided that $\Theta(S)$ satisfies the above properties, it can be used to define directly the on-shell time-ordered product of order two. It will, however, be more convenient to relate $\Theta(S)$ to the map considered in \cite{onshell} which gives the connection between between the on-shell and off-shell time-ordered products of order $k$, cf.  Theorem \ref{thm:chion} and \ref{thm:brouduetsch}.

\begin{theorem}\label{thm:onshell}Let $Q:\cD'(\R^n)\to\cD'(\R^n)$ be a nonzero differential operator of order $q$ with constant coefficients and assume $v\in\cD'(\R^n)$ satisfies $\deg v=-q$ and $Q v=c\delta$ for some $c\in\C$. Assign to a differential operator $S$ of order $s$ with constant coefficients the distribution
\[
\Theta(S):= p_s \big( (Q|_s)^* Q \big) S v\in\cD'(\R^n),
\]
Then the mapping $S\mapsto \Theta(S)$ is linear. Moreover, $\Theta(S Q)=0$ for each differential operator $S$ with constant coefficients. 
\end{theorem}
\proof For linearity, we have to prove that if $s'\geq s$ then $ p_{s'} \big( (Q|_{s'})^* Q \big) S v$ equals  $p_s \big( (Q|_s)^* Q \big) S v$. 
Because $Q$ has no nonzero solutions in $\Do$, this is equivalent to
\[
Q\left[p_{s'} \big( (Q|_{s'})^* Q \big) -p_s \big( (Q|_{s})^* Q \big) \right] S v=0.
\]
One has
\begin{eqnarray*}
Q\left[p_{s'} \big( (Q|_{s'})^* Q \big) -p_s \big( (Q|_{s})^* Q \big) \right]S v\\= \left[p_{s'} \big( Q (Q|_{s'})^*  \big) -p_s \big(Q (Q|_{s})^*  \big) \right] S Q v \\= c \left[p_{s'} \big( Q (Q|_{s'})^*  \big) -p_s \big( Q (Q|_{s})^*  \big) \right] S \delta.
\end{eqnarray*}
To prove that the last expression vanishes it suffices to show $p_{s'} \big( Q (Q|_{s'})^*  \big)$ and $p_s \big( Q (Q|_{s})^*  \big)$ coincide on $\Do_{\leq s}$. Indeed, as operators from $\Do_{\leq s+q}$ to $\Do_{\leq s}$, these are orthogonal projections to the same space, as follows from Lemma \ref{lem:adjoint}.

To prove $\Theta(S Q)=0$, let us remark that this is equivalent to $Q\Theta(S Q)=0$, as $Q$ has no nonzero solutions in $\Do$. We have 
\[
Q\Theta(S Q)= Q p_{s+q} \big( (Q|_{s+q})^* Q \big) S Qv = p_{s+q} \big( Q (Q|_{s+q})^*  \big) Q S\delta.
\]
Clearly, $Q S\delta\in\Ran(Q|_{s+q})=(\ker(Q|_{s+q})^*)^{\bot}$, so it is projected out by $p_{s+q} \big( Q (Q|_{s+q})^*  \big)$.
\qed\\

The next lemma gives the connection between the map $S\mapsto \Theta(S)$ and a generalized version of the map $\chi$ from Theorem \ref{thm:chion} and \ref{thm:brouduetsch}.

\begin{lemma}\label{lemma:chi1}Let $Q$ and $v$ be as in Theorem \ref{thm:onshell}. There is a linear map $S\mapsto\chi(S)$ on the space of differential operators with constant coefficients, s.t. $\ord\,\chi(S)\leq \ord\,S $ and $\Theta(S)=\chi(S)v$ for all $S$.
\end{lemma}
\proof Since $\Theta(S)-Sv\in\Do_{\leq s-q}$, it can be written as $\chi_1(S)\delta$ for some differential operator $\chi_1(S)$ with constant coefficients of order $s-q$. More precisely, $\chi_1(S)=\cS_{s-q}\left[ \Theta(S)-Sv\right](-\partial)$, hence the assignment $S\mapsto\chi_1(S)$ is linear. The map $S\mapsto\chi(S):=\chi_1(S)Q+S$ satisfies the required properties.
\qed\\

Setting $Q=\Box+m^2$ and $v=\Delta_{\rm F}$, it follows automatically that $\chi$ satisfies the conditions given in Theorem \ref{thm:chion} except Lorentz covariance. But it easy to see that $\chi$ is defined purely using Lorentz covariant quantities.

\begin{corollary}Let $Q=\Box+m^2$, $v=\Delta_{\rm F}$, let $\Theta(S)$ be as in Theorem \ref{thm:onshell} and $\chi$ as in Lemma \ref{lemma:chi1}. Then $\chi$ satisfies the conditions listed in Theorem \ref{thm:chion}.
\end{corollary}

The adjoint of $(\Box+m^2)|_r$ is easily computed from Lemma \ref{lem:adjoint}. One has $\left[(\Box+m^2)|_r\right]^*=P_r\left[(x_{\mu}x^{\mu}+m^2)|_{r+2}\right]$, where $P_r$ is the orthogonal projection to $\Dor$. In particular if $m=0$, since $x_{\mu}x^{\mu}$ maps $\Do_{\leq r+2}$ to $\Dor$, this simplifies to $(\Box|_r)^*=x_{\mu}x^{\mu}|_{r+2}$. Consequently the distribution $\Theta(S)$ from Theorem \ref{thm:onshell} equals
\[
\Theta(S)=p_s(x_{\mu}x^{\mu} \Box)S\Delta_{\rm F}.
\]
(Without the obligation to restrict $x_{\mu}x^{\mu}$ to a subspace of $\Do$).

\section{Outlook}\label{sec:outlook}

We have set up a unified formalism to treat the different renormalisation conditions which appear in quantum field theory. 

From the mathematical point of view, our analysis confirms once more that Steinman's scaling degree is a very natural notion in the problem of extension of distributions. We expect to find interesting applications in the various problems in singular analysis of partial differential equations, where for technical reasons one has to work with distributions on $\R^n\setminus\{0\}$ rather than $\R^n$. To this end, one would first need to make the connection between distributions of specific scaling degree and  spaces of distributions used in microlocal singular analysis.

On a separate note, it is natural to ask whether a generalization of our results to operator-valued distributions is possible, at least in the case of bounded operators. Such kind of result would be helpful in establishing a position-space approach to renormalisation which does not refer to Wick's theorem to reduce the problem to ordinary distributions.

\end{document}